\documentclass[11pt]{article}
\usepackage{amsmath,amsthm,amssymb}
\usepackage{algorithm,algorithmic}
\usepackage{enumerate}
\usepackage{hyperref}
\usepackage{tikz}
\usepackage[mathlines]{lineno}
\usepackage{dsfont} 
\usepackage{soul} 
\usepackage{color}
\usepackage{booktabs}
\usepackage{nicefrac} 
\usepackage{graphicx}
\usepackage[numbers]{natbib}

\numberwithin{figure}{section}   
\numberwithin{table}{section}   
\numberwithin{equation}{section}  
\numberwithin{algorithm}{section}   

\theoremstyle{plain}
\newtheorem{thm}{Theorem}[section]
\newtheorem{lem}[thm]{Lemma}

\newtheorem{cor}[thm]{Corollary}

\theoremstyle{definition}
\newtheorem{ex}[thm]{Example}
\newtheorem{rem}[thm]{Remark}
\newtheorem{defn}[thm]{Definition}

\definecolor{red}{rgb}{.8,0,0}

\definecolor{bblu}{rgb}{0,0,1}

\definecolor{gray}{rgb}{.4,.4,.4}

\definecolor{gre}{rgb}{0,.7,0}

\newcommand{\Z}{{\mathbb Z}}

\newcommand{\C}{{\mathbb C}}

\newcommand{\powerset}{\raisebox{.15\baselineskip}{\Large\ensuremath{\wp}}}

\usepackage[bb=boondox]{mathalfa} 

\newcommand{\F}{\mathcal F}
\newcommand{\G}{\mathcal G}

\newcommand{\bit}{\begin{itemize}}
\newcommand{\eit}{\end{itemize}}
\newcommand{\ben}{\begin{enumerate}}
\newcommand{\een}{\end{enumerate}}
\newcommand{\beq}{\begin{equation}}
\newcommand{\eeq}{\end{equation}}
\newcommand{\bea}{\begin{eqnarray*}}
\newcommand{\eea}{\end{eqnarray*}}
\newcommand{\bpf}{\begin{proof}}
\newcommand{\epf}{\end{proof}}

\newcommand{\lp}{\!\left(}
\newcommand{\rp}{\right)}

\DeclareMathOperator*{\rank}{rank}
\DeclareMathOperator*{\dsum}{\bigoplus}

\begin{document}

\title{Symbolic Listings as Computation}
\author{Hamilton Sawczuk, Edinah Gnang \thanks{Department of Applied Mathematics and Statistics, Johns Hopkins University, Baltimore, MD 21218, USA (egnang1@jhu.edu)}} 

\maketitle\vspace{-15pt}

\begin{abstract} 
We propose an algebraic model of computation which formally relates symbolic listings, complexity of Boolean functions, and low depth arithmetic circuit complexity. In this model, algorithms are arithmetic formulae expressing symbolic listings of YES instances of Boolean functions, and computation is executed via partial differential operators. We consider the Chow rank of an arithmetic formula as a measure of complexity and establish the Chow rank of multilinear polynomials with totally non-overlapping monomial support. We also provide Chow rank non-decreasing transformations from sets of graphs to sets of functional graphs.

\end{abstract}



\section{Introduction}\label{s:intro}

In his epochal title, “An Investigation of the Laws of Thought,” George Boole laid the foundations of Boolean algebra, the first pillar of the theory of computation \cite{Boo09}. In their own seminal works, Godel, Church, and Turing formalized the notion of effective computability via the theories of general recursive functions, $\lambda$-calculus, and Turing machines respectively \cite{God65, Chu23, Tur37}. These formulations were shown to be equivalent, and the Turing machine endures as the second pillar of the theory of computation. Finally, in his 1940 master's thesis, Claude Shannon described a general procedure for implementing Boolean algebra via switching circuits, the third pillar of the theory of computation, which would eventually usher in the information era \cite{Sha38}.

Although we believe $P\neq NP$, decades of research have produced no super\textendash polynomial lower bounds on Turing machine complexity. In the 1970s and 80s there was some hope that analyzing circuits would prove more fruitful. It is widely believed that $NP$ does not admit circuits of polynomial size, and we know that a proof of this would imply $P\neq NP$. Unfortunately, lower bounds seem difficult to show for general circuits as well \cite{Raz97}. This in turn led researchers to consider restricted classes of circuits, for which some non\textendash trivial lower\textendash bounds have been proven \cite{Ajt83, Wil11}. Inspired by Boole's translation of logic into algebra and the authors' interest in symbolic listings, this work proposes an algebraic model of computation called “differential computers." We think of differential computers as an algebraic implementation of a restricted class of Turing machines, and the programs of differential computers are closely related to low-depth arithmetic circuit complexity. Thus, differential computers lie at the intersection of the three pillars of the theory of computation and symbolic listings.

In the model of differential computers, computation is carried out by the action of differential operators on particular polynomials whose monomial support corresponds to YES instances of a Boolean function. The Chow rank of the polynomial underlying a differential computer is of particular interest, and can be seen as a measure of the compressibility of the truth table of $\F$. By using the Chow rank as a measure of complexity, this work also relates aspects of low depth arithmetic circuit complexity to the complexity of Boolean functions.

Recent depth reduction results motivate our focus on low depth arithmetic circuits \cite{Val81, Raz13, Gup16, Val79, Agr08}. Differential operators have previously been used in the context of arithmetic complexity by Baur and Strassen \cite{Bau83}. More recently, Brand and Pratt matched the runtime of the fastest known deterministic algorithm for detecting subgraphs of bounded path–width using a method of partial derivatives \cite{Bra20}. We refer the reader to the excellent surveys on partial differential methods in arithmetic complexity by Shpilka and Yehudayoff as well as Chen et al.
\cite{Shp10, Che11}. The importance of the computational power of differential operators is
reinforced by the role of these operators in training machine learning models.

\section{A Symbolic Model of Computation}

\subsection{Symbolic Listings}

Throughout the paper we work over the field $\mathbb{C}$. Let $\mathbb{Z}_{n}$ denote the set formed
by the first $n$ consecutive non--negative integers, i.e.
\[
    \mathbb{Z}_{n}:=\mathbb{Z}\cap\left[0,n\right).
\]
Let $G$ be an unweighted, directed graph on $n$ vertices with self\textendash loops. We will primarily specify the data of $G$ by a standard binary adjacency matrix $B_G$, although the following two representations will be critical for an intuitive understanding of symbolic listings. The first is the \textbf{symbolic adjacency matrix} $A_G$, where
\[
    A_G[i,j]=
    \begin{cases} 
        a_{i,j} & (i,j) \in E(G) \\
        0       & (i,j) \not\in E(G)
    \end{cases}.
\]
The second is the \textbf{monomial edge listing} $M_G$, where
\[
    M_G=\prod_{(i,j)\in E(G)}a_{i,j}.
\]
We will also use $A$ to denote the matrix of symbolic variables where $A[i,j]=a_{i,j}$ for all $\lp i,j \rp \in \Z_n \times \Z_n$. This can be thought of as the symbolic adjaceny matrix of the \textbf{totally complete graph} ($\mathbb{K}_n$) in which all possible directed and loop edges are present.

We say a graph $G$ is \textbf{functional} if every vertex in $G$ has out\textendash degree equal to one. In this case we think of $G$ as representing a function $g:\mathbb{Z}_n\rightarrow\mathbb{Z}_n,$ and write
\[
    M_g:=M_G=\prod_{i\in\Z_n}a_{i,g(i)}.
\]
Observe that the polynomial
\[
    P_{\Z_n\rightarrow\Z_n}(A):=\sum_{f:\Z_n\rightarrow\Z_n}\prod_{i\in\Z_n}a_{i,f(i)}=\sum_{f:\Z_n\rightarrow\Z_n}M_f
\]
is the sum of monomial edge listings of all functional graphs on $n$ vertices. Thus to check if a graph $G$ is functional, it suffices to check if $M_G$ appears as a monomial in $P_{\Z_n\rightarrow\Z_n}$.

\begin{defn}
    Given $m\in\Z_{>0}$ and a Boolean function $\F:\{0,1\}^n\rightarrow\{0,1\}$, an \textbf{additive listing} is a polynomial $P_{\F,m}(\boldsymbol{a})$. The coefficients of $P_{\F,m}(\boldsymbol{a})$ are $m$th roots of unity and its monomial support in the variables $a_0,...,a_{n-1}$ corresponds to the set of YES instances of $\F$. Explicitly,
    \[
        P_{\F,m}(\boldsymbol{a})=\sum_{
        \begin{array}{c}
            \boldsymbol{b}\in\{0,1\}^n \\
            \F\left(\boldsymbol{b}\right)=1
        \end{array}}
        \omega_{\text{lex}(\boldsymbol{b})}\prod_{i\in\Z_n}a_i^{b_i}, \label{eq:RE}
    \]
    where $\text{lex}:\{0,1\}^n\rightarrow\Z_{2^n}$ is any enumeration of binary $n$\textendash vectors, and for all $i\in\Z_{2^n}$, $\omega_i^m=1$. If $m=1$, we omit the subscript $m$ from $P_{\F,m}$ and refer to the $P_{\F}$ as a binary additive listing. To see that $P_{\F,m}(\boldsymbol{a})$ contains the data of the truth table of $\F$, notice that a monomial term appears in $P_{\F,m}$ exactly when it corresponds to an assignment $\boldsymbol{b}$ of $n$ variables such that $\F(\boldsymbol{b})=1.$ That is, $P_{\F,m}(\boldsymbol{a})$ is a symbolic \textbf{additive listing} of monomials corresponding to YES instances of $\F$.
\end{defn}

We note that it is critical to allow for non\textendash binary coefficients in additive listings to separate decision problems from their counting analogues as discussed in Examples \ref{ex:gi} and \ref{ex:perm}. On the other hand, arbitrary coefficients would introduce unbounded computational power into our model. The choice of roots of unity as coefficients is inspired by the distinction between the determinant and the permanent, also discussed in Example \ref{ex:perm}.

\begin{rem} We point out that the binary additive listing $P_{\F}$ is closely related to the Lagrange interpolating polynomial of $\F$. The Lagrange interpolant of $\F$ in the variables $y_0,\cdots,y_{n-1}$ is
\[
    L_{\F}\left(\boldsymbol{y}\right)=\sum_{
    \begin{array}{c}
        \boldsymbol{b}\in\{0,1\}^n \\
        \F(\boldsymbol{b})=1
    \end{array}}
    \prod_{i\in\Z_n}\left(\frac{y_i-(1-b_i)}{2b_i-1}\right),
\]
By construction, $L_{\F}$ is the unique minimum degree polynomial whose evaluations coincides with $\F$ at each point on the Boolean hypercube. $P_\F$ is the canonical representative
\[
    P_{\F}(\boldsymbol{a})\equiv L_{\F}\left(\boldsymbol{y}\right)\mod\left\{ \frac{y_i-\left(1-b_i\right)}{2b_i-1}-a_i^{b_i}:
    \begin{array}{c}
        b_i\in\{0,1\} \\
        i\in\Z_n
    \end{array}
    \right\}. \label{eq:interp_reduc}
\]
This binomial reduction simply corresponds to the substitutions
\[
    \left\{ \frac{y_i-(1-b_i)}{2b_i-1}\leftarrow\left(a_i\right)^{b_i}:
    \begin{array}{c}
        b_i\in\{0,1\} \\
        i\in\Z_n
    \end{array}\right\},
\]
which move the data stored in the truth table of $\F$ from evaluations of the interpolating polynomial $L_{\F}(\boldsymbol{y})$ to coefficients of the additive listing $P_{\F}(\boldsymbol{a})$. Although $P_\F$ is unique, for $m\in\Z_{\geq2}$ there are $m^{\left|\left\{ \F(\boldsymbol{b})=1:\,\boldsymbol{b}\in\left\{ 0,1\right\} ^{n}\right\} \right|}$ distinct additive listings $P_{\F,m}$.
\end{rem}

Given a graph $G$, we can check if the monomial encoding $M_G$ appears as a term in $P_{\Z_n\rightarrow\Z_n}$ by taking a sequence of partial derivatives and evaluating the resulting polynomial at zero. Specifically, the expression
\[
    F_{\Z_n\rightarrow\Z_n}(G)=\left.\left(\prod_{(i,j)\in E(G)}\frac{\partial}{\partial a_{ij}}\right)P_{\Z_n\rightarrow\Z_n}(A)\right\rfloor_{A=\mathbb{0}}
\]
evaluates to one exactly when $M_G$ appears as a monomial term in $P_{\Z_n\rightarrow\Z_n}$, i.e. when $G$ is a functional graph. According to the definition we introduce formally in Section \ref{ss:dc}, we refer to $F_{\Z_n\rightarrow\Z_n}$ as a \textbf{differential computer}.

\begin{ex}\label{ex:funct_list}
    Let $n=2$. Then the additive listing of functional graphs on two vertices is
    \[
        P_{\Z_2\rightarrow\Z_2}(A)=\color{magenta}{ a_{00}a_{10}}+\color{violet}{ a_{00}a_{11}}+\color{blue}{a_{01}a_{10}}+\color{teal}{a_{01}a_{11}}\color{black}.
    \]
    \begin{figure}\begin{center}$\begin{array}{cccc}
        \begin{tikzpicture}
            \node (0) at (0,0) {};
            \node (1) at (2,0) {};
            \draw[fill=black] (0,0) circle (3pt);
            \draw[fill=black] (2,0) circle (3pt);
            \node at (0,-0.5) {$0$};
            \node at (2,-0.5) {$1$};
            \draw (0) edge[color=magenta,very thick,->,out=135,in=45,looseness=10] (0) node [above=20pt,fill=white] {\textcolor{magenta}{$a_{00}$}};
            \draw (1) edge[color=magenta,very thick,->,out=225,in=315,looseness=1] (0) node [below=22pt,left=16pt,fill=white] {\textcolor{magenta}{$a_{10}$}};
        \end{tikzpicture} &
        \begin{tikzpicture}
            \node (0) at (0,0) {};
            \node (1) at (2,0) {};
            \draw[fill=black] (0,0) circle (3pt);
            \draw[fill=black] (2,0) circle (3pt);
            \node at (0,-0.5) {$0$};
            \node at (2,-0.5) {$1$};
            \draw (0) edge[color=violet,very thick,->,out=135,in=45,looseness=10] (0) node [above=20pt,fill=white] {\textcolor{violet}{$a_{00}$}};
            \draw (1) edge[color=violet,very thick,->,out=135,in=45,looseness=10] (1) node [above=20pt,fill=white] {\textcolor{violet}{$a_{11}$}};
            \draw (0) edge[color=white,out=315,in=225,looseness=1] (1) node [below=22pt,right=20pt,fill=white] {\textcolor{white}{$a_{0}$}};
        \end{tikzpicture} &
        \begin{tikzpicture}
            \node (0) at (0,0) {};
            \node (1) at (2,0) {};
            \draw[fill=black] (0,0) circle (3pt);
            \draw[fill=black] (2,0) circle (3pt);
            \node at (0,-0.5) {$0$};
            \node at (2,-0.5) {$1$};
            \draw (0) edge[color=blue,very thick,->,out=45,in=135,looseness=1] (1) node [above=22pt,right=16pt,fill=white] {\textcolor{blue}{$a_{01}$}};
            \draw (1) edge[color=blue,very thick,->,out=225,in=315,looseness=1] (0) node [below=22pt,left=16pt,fill=white] {\textcolor{blue}{$a_{10}$}};
        \end{tikzpicture} &
        \begin{tikzpicture}
            \node (0) at (0,0) {};
            \node (1) at (2,0) {};
            \draw[fill=black] (0,0) circle (3pt);
            \draw[fill=black] (2,0) circle (3pt);
            \node at (0,-0.5) {$0$};
            \node at (2,-0.5) {$1$};
            \draw (0) edge[color=teal,very thick,->,out=315,in=225,looseness=1] (1) node [below=22pt,right=16pt,fill=white] {\textcolor{teal}{$a_{11}$}};
            \draw (1) edge[color=teal,very thick,->,out=135,in=45,looseness=10] (1) node [above=20pt,fill=white] {\textcolor{teal}{$a_{01}$}};
        \end{tikzpicture}
    \end{array}$\end{center}
    \caption{Pictorial listing of functional graphs on two vertices.}
    \end{figure}
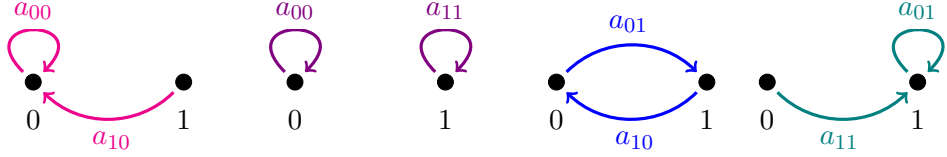\noindent
    To check if the graph with monomial edge listing $a_{00}a_{01}$ is functional, we perform the following sequence of partial derivatives followed by an evaluation at zero of all variables in the resulting polynomial.
    \[
        \left.\frac{\partial}{\partial a_{00}}\frac{\partial}{\partial a_{01}}P_{\Z_2\rightarrow\Z_2}(A)\right\rfloor_{A=\mathbb{0}}
        =\left.\frac{\partial}{\partial a_{00}}\frac{\partial}{\partial a_{01}}\left(\color{magenta}{a_{00}a_{10}}+\color{violet}{a_{00}a_{11}}+\color{blue}{a_{01}a_{10}}+\color{teal}{ a_{01}a_{11}}\right)\right\rfloor_{A=\mathbb{0}}
    \]
    \[
        =\left.\frac{\partial}{\partial a_{00}}\left({\color{blue}a_{10}} + {\color{teal}a_{11}} \right)\right\rfloor_{A=\mathbb{0}}
        =0
    \] \noindent
    Thus said graph is not functional, which can also be seen by observing that vertex $0$ has out\textendash degree two.
\end{ex}

\subsection{The Chow Rank}\label{ss:chow}

Although the binary additive listing $P_\F$ is unique (by the uniqueness of the Lagrange interpolant), our interest lies in the space of Chow decompositions of $P_{\F,m}.$ Chow decompositions will provide us with a measure of the complexity of $P_{\F,m}$ and in turn the Boolean function $\F$ that it encodes.

\begin{defn}
    Recall that the \textbf{Chow rank}, also called the product or split rank, of the degree $d-1$ polynomial $P\in\mathbb{C}\left[x_{0},...,x_{n-1}\right]$  is the smallest number of additive terms $\rho$ in an expression of the form
    \[
        P(\boldsymbol{x})=\sum_{u\in\Z_\rho} \prod_{v\in\Z_d} l_{u,v}\lp x_{0},...,x_{n-1}\rp
    \]
        where each $l_{u,v}$ is a non-homogeneous linear form in the variables of $x_{0},...,x_{n-1}$ \cite{Ilt16}. Note that the coefficients of the linear forms $l_{u,v}$ can be taken from a third order hypermatrix $H \in \mathbb{C}^{\rho \times d \times (n+1)}$ such that
    \[
        P(\boldsymbol{x})=\sum_{u\in\Z_\rho} \prod_{v\in\Z_d} \left(H_{u,v,n} + \sum_{w\in\Z_n} H_{u,v,w} x_w \right).
    \]
    We will refer to the number of Chow rank one summands in a Chow decomposition underlied by the hypermatrix $H$ as $\rho(H)$.
\end{defn}

First, notice that the expanded form of $P$ expresses a Chow decomposition in which each variables is seen as a linear form with a single non\textendash zero coefficient. Thus, the number of terms in the expanded form of $P$ is an upper bound on its Chow rank. On the other hand, the Chow rank yields a lower bound on the fan\textendash in of the top addition gate of a depth\textendash 3 $\sum \prod \sum$ arithmetic circuit expressing $P$, as shown in Figure \ref{fig:sps_formula} \cite{Gup16}.

\begin{figure}[ht]
    \centering
    \begin{tikzpicture}
[  
   level 2/.style={sibling distance=25mm},
   level 3/.style={sibling distance=8mm},
   level 4/.style={sibling distance=10mm,level distance=10mm}]
    \node (P){$P$}
    child{node (top) {$\bigoplus$}  edge from parent [<-]
        child {node (times-1){$\bigotimes$} edge from parent [<-]
            child {node (plus-1) {$\bigoplus$} edge from parent [<-]
                child {node (variables) {$c_0x_0$}}
                child {node {$\cdots$}}
                child {node {$c_{n-1}x_{n-1}$}}
            }
            child {node {$\bigoplus$} edge from parent [<-]
            }
            child {node {$\bigoplus$} edge from parent [<-]
            }
        }
        child {node {$\bigotimes$} edge from parent [<-]
            child {node {$\bigoplus$}}
            child {node {$\bigoplus$}}
            child {node {$\bigoplus$}}
        }
        child {node {$\bigotimes$} edge from parent [<-]
            child {node {$\bigoplus$}}
            child {node {$\bigoplus$}}
            child {node {$\bigoplus$}}
        };
    \draw[-,dashed] (top) -- +(-6,0) node[left] { The top addition gate};
    \draw[-,dashed] (times-1) -- +(-3.5,0) node[left] { Multiplication gate};
    \draw[-,dashed] (plus-1) -- +(-2.7,0) node[left] { Addition gate};
    \draw[-,dashed] (variables) -- +(-1.7,0) node[left] { Inputs};
    \draw[-,dashed] (-1.5,-0.75) -- +(-4.5,0) node[left] { Fan\textendash in: number of inputs to the top gate};
    \draw[-,dashed] (P) -- +(-6,0) node[left] { The target polynomial};
    };
\end{tikzpicture}
    \caption{$\sum \prod \sum$\textendash formula.}
    \label{fig:sps_formula}
\end{figure}
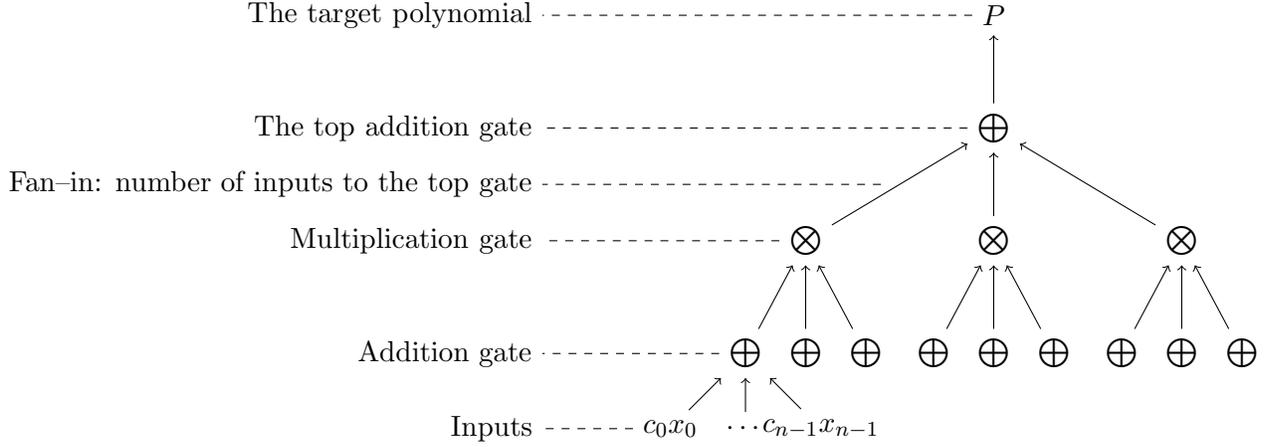

\begin{ex}
    Observe that the additive listing $P_{\Z_n\rightarrow\Z_n}$ exhibits an exponential separation between its Chow rank and the Chow rank upper bound given by its expanded form
    \[
    P_{\Z_n\rightarrow\Z_n}(A)=\sum_{f:\Z_n\rightarrow\Z_n}\prod_{i\in\Z_n}a_{i,f(i)}=\prod_{i\in\Z_n}\sum_{j \in \Z_n}a_{i,j}.
    \]
\end{ex}

\subsection{Differential Computers} \label{ss:dc}

We now describe formally how additive listings can be used to implement Boolean functions. Subsequent examples aim to illustrate this implementation and motivate our definition of additive listings.

\begin{defn}
    Given an additive listing $P_{\mathcal{F},m}$, a \textbf{differential computer} is the following implementation of the Boolean function $\F:\{0,1\}^n\rightarrow\{0,1\}$
    \[
       F_{\F}(\boldsymbol{b})=\lp\left.\left(\underset{i\in\Z_{n}}{\prod}\left(\frac{\partial}{\partial a_{i}}\right)^{b_{i}}\right)P_{\F,m}(\boldsymbol{a})\right\rfloor _{\boldsymbol{a}=\boldsymbol{0}}\rp^{m}.
    \]
    In the notation $\left(\frac{\partial}{\partial a_{i}}\right)^{b_{i}}$, $b_i$ is an exponent indicating that the partial derivative with respect to $a_i$ should be taken exactly when $b_i=1$. We refer to $m$ as the exponent parameter, which maps the output of the inner expression to $\{0,1\}$. In analogy with a classical computer, $P_\F$ is the program, and $\boldsymbol{b}$ is the input. The computation is executed via the sequence of partial derivatives followed by the evaluation of the resulting polynomial at $\boldsymbol{a}=\boldsymbol{0}$. Finally, if the hypermatrix $H\in\mathbb{Q}^{\rho\times d\times(n+1)}$ underlies $P_\F$ then the complexity of $F_\F$ is said to be $|H| = (n+1)\rho d$. Thus the Chow rank of $F_\F$ is a lower bound on its optimal running time.
\end{defn}

\begin{ex}
    Given $S\subset\Z_n$, define $\F_{=S}$, $\F_{\subseteq S}:\{0,1\}^n\rightarrow\{0,1\}$ which test if an input set $T$, represented by the indicator vector $\boldsymbol{1}_T$, is equal to or a subset of the fixed set $S$. That is,
    \[
    \F_{=S}(\boldsymbol{1}_T)=
    \begin{cases}
        1   &   T=S \\
        0   &   T\neq S
    \end{cases},\text{ and   }
    \F_{\subseteq S}(\boldsymbol{1}_T)=
    \begin{cases}
        1   &   T\subseteq S \\
        0   &   T\not\subseteq S
    \end{cases}.
    \]
    We can write binary additive listings for these functions as follows.
    \[
        P_{=S}(\boldsymbol{a})=\prod_{i\in S}a_i
    \]
    \[
        P_{\subseteq S}(\boldsymbol{a})=\sum_{R\subseteq S}\prod_{i\in R}a_i=\prod_{i\in S}( 1+a_i)
    \]
    This illustrates once again that although binary additive listings are unique, there may be an exponential gap between their Chow rank and $\rho(H)$ for $H$ underlying a given Chow decomposition. Since both of the right hand side Chow decompositions satisfy $\rho(H)=1$, they are trivially optimal. Now we can use these optimal additive listings to construct optimal differential computers as follows.
    \[
        F_{=S}(\boldsymbol{1}_T)=\left.\left(\prod_{i\in T}\frac{\partial}{\partial a_i}\right)\prod_{i\in S}a_i\right\rfloor_{\boldsymbol{a}=\boldsymbol{0}}
    \]
    \[
        F_{\subseteq S}(\boldsymbol{1}_T)=\left.\left(\prod_{i\in T}\frac{\partial}{\partial a_i}\right)\prod_{i\in S}(1+a_i)\right\rfloor_{\boldsymbol{a}=\boldsymbol{0}}
    \]
\end{ex}

\begin{ex} \label{ex:gi}
    Building upon the previous example, consider $\F_{\sim= G}$ which tests if an input graph $H$ is isomorphic to the fixed graph $G$. Differential computers can equivalently be defined with respect to matrix inputs as follows.
    \[
    F_{\sim=G}\left(B_{H}\right)=\left.\left(\underset{\lp i,j\rp \in\Z_{n}\times\Z_{n}}{\prod}\left(\frac{\partial}{\partial a_{i,j}}\right)^{B_{H}[i,j]}\right)P_{\sim=G}\left(A\right)\right\rfloor _{A=\mathbb{0}}
    \]
    We can write the expanded form of the binary additive listing $P_{\sim=G}$ as
    \[
        P_{\sim=G}(A)=\sum_{\sigma\in\text{S}_n/\text{Aut}(G)}\prod_{(i,j)\in E(\sigma G \sigma^{-1})}a_{i,j},
    \]
    which underlies the differential computer
    \[
        F_{\sim=G}\left(B_{H}\right)=\left.\left(\underset{\lp i,j\rp \in\Z_{n}\times\Z_{n}}{\prod}\left(\frac{\partial}{\partial a_{i,j}}\right)^{B_{H}[i,j]}\right)\sum_{\sigma\in\text{S}_n/\text{Aut}(G)}\prod_{(i,j)\in E(\sigma G \sigma^{-1} )}a_{i,j}\right\rfloor_{A=\mathbb{0}}.
    \]
    Observe that evaluating $P_{\sim=C_n}(B_H)$ counts the number of Hamiltonian cycles in the graph $H$, which is known to be \#P\textendash Complete \cite{Val79}. Since Chow decompositions of $P_{\sim=C_n}$ correspond to depth-three arithmetic formula computing $P_{\sim=C_n}$, we conjecture that the Chow rank of $P_{\sim=C_n}$ is super-polynomial in $n$.
\end{ex}

\begin{ex} \label{ex:perm}
    Consider $\F_{\in\text{S}_n}:\{0,1\}^{n\times n}\rightarrow\{0,1\}$ which tests if matrix represents a permutation. Recall that
    \[
        \text{Per}(A)=\sum_{\sigma\in\text{S}_n}\prod_{i\in\Z_n}a_{i,\sigma(i)},
    \]
    and notice that $\text{Per}(A)$ is the sum of all monomial edge listings of permutations on $\Z_n$. Thus, $\text{Per}(A)$ is the unique binary additive listing of permutations on $\Z_n$ and underlies a differential computer $F_{\in\text{S}_n}$ with exponent parameter one. At the same time, $\text{Per}(B_G)$ computes the number of spanning subgraphs of $G$ which are the disjoint union of directed cycles, another $\#P$\textendash Complete problem \cite{Val79}. In fact, it is known that
    \[
        \text{Chow rank}(\text{Per}(A))\geq \frac{2^n}{\sqrt{n}},
    \]
    and thus it must be hard to test bijectivity on a differential computer with exponent parameter one \cite{Ilt16}. On the other hand, we know that testing bijectivity is equivalent to testing injectivity and surjectivity, both of which are easy in the classical model of computation. This issue is resolved by allowing additive listings with non\textendash zero coefficients different from one. Recalling
    \[
        \text{Det}(A)=\sum_{\sigma\in\text{S}_n}\text{sgn}(\sigma)\prod_{i\in\Z_n}a_{i,\sigma(i)},
    \]
    notice that $P_{\in\text{S}_n,2}=\text{Det}(A)$ is an additive listing for $\F_{\in\text{S}_n}$ with exponent parameter two, bypassing the exponential Chow rank lowerbound on $\text{Per}(A)$.
\end{ex}

\begin{rem}
    Given a matrix $A\in\C^{n\times n}$, recall the identity
    \[
        A^{-1}=\nabla_{A^\top }\text{ln}(\text{det}(A))
    \]
    where for all $f:\C^{n\times n}\rightarrow \C$, $\nabla_{A^\top }f(A)$ is defined entry\textendash wise as
    \[
        (\nabla_{A^\top }f(A))[i,j]=\frac{\partial}{\partial a_{j,i}}f(A).
    \]
    This expression can be seen as $n^2$ parallel differential computers where the evaluation at $A=\mathbb{0}_{n\times n}$ is skipped and our model is enriched with a natural logarithm gate.
    Further, by the identity 
    \[
        \begin{pmatrix}
            I_n                     & X     &   \mathbb{0}_{n\times n}  \\
            \mathbb{0}_{n\times n}  & I_n   &   Y                       \\
            \mathbb{0}_{n\times n}  & \mathbb{0}_{n\times n} & I_n
        \end{pmatrix}^{-1}
        =
        \begin{pmatrix}
            I_n                     & -X    &   XY  \\
            \mathbb{0}_{n\times n}  & I_n   &   -Y  \\
            \mathbb{0}_{n\times n}  & \mathbb{0}_{n\times n} & I_n  
        \end{pmatrix},
    \]
    asymptotically optimal algorithms for inverting matrices yield asymptotically optimal algorithms for multiplying matrices. Thus we hope that greater understanding of the computational power of derivative operators may offer novel insight into the complexity of matrix multiplication.
\end{rem}

\subsection{Functional Computers}

When considering Boolean functions on graphs, it can simplify matters greatly to assume that these graphs are functional. In Lemma \ref{lem:reduc}, we provide some justification for restricting our attention to functional graphs. We then conclude by discussing how the homogeneity of functional graphs translates to simpler differential computers.

We begin with the observation that the binary adjacency matrix of a graph $G$ on $n$ vertices, $B_G$, can itself be viewed as a function $g:\Z_n\times\Z_n\rightarrow\{0,1\}$, which maps an ordered pair $(i,j)$ to the matrix entry $B_{G}[i,j]$. In this way binary adjacency matrices exhibit a bijection between graphs and functions $g:\Z_n\times\Z_n\rightarrow\{0,1\}$. Then by identifying both the domain $\Z_n\times\Z_n$ and the codomain $\Z_2$ of $g$ either as subsets of $\Z_{n^2+2}$ or simply $\Z_{n^2}$, we produce injections from the set of graphs on $n$ vertices to the set of functional graphs on $n^2+2$ or $n^2$ vertices.

\begin{lem} \label{lem:reduc}
    Let $\G_n$ and $\Tilde{\G}_n$ denote the set of all graphs and functional graphs respectively on $n$ vertices. Then there exist two efficient transformations $T_f:\powerset(\G_n)\rightarrow\powerset(\Tilde{\G}_{n^2+2})$ and $T:\powerset(\G_n)\rightarrow\powerset(\Tilde{\G}_{n^2})$ which are Chow rank non\textendash decreasing.
\end{lem} \label{lem:trans}
\begin{proof}
    Consider a graph $G$ on $n$ vertices with monomial edge listing $M_G$. Let $f:\Z_2\rightarrow\Z_2$, and define $T_f$ by its action on $M_G$
    \[
        \prod_{(i,j)\in E(G)}a_{i,j} \rightarrow \left(M_{f}\prod_{(i,j)\in E(G)}x_{2+ni+j,1}\prod_{(i,j)\in \overline{E(G)}}x_{2+ni+j,0}\right).
    \]
    Similarly, define $T$ by
    \[
        \prod_{(i,j)\in E(G)}a_{i,j} \rightarrow \left(\prod_{(i,j)\in E(G)}x_{ni+j,1}\prod_{(i,j)\in \overline{E(G)}}x_{ni+j,0}\right).
    \]
    
    Observe first that both $T_f(M_G)$ and $T(M_G)$ contain exactly the data of the listing $M_G$. Next, notice that $T_f(M_G)$ and $T(M_G)$ are monomial edge listings of functional graphs on $n^2+2$ and $n^2$ vertices respectively. Thus, given a set of graphs $S$ on $n$ vertices we can transform $S$ into a set of functional graphs $T_f(S)$ or $T(S)$ by applying the the appropriate transformation to the monomial edge listing of each graph in $S$. Now we consider the Chow ranks of
    \[
        P_{\in S}=\sum_{G\in S}\prod_{(i,j)\in E(G)}a_{i,j},
    \]
    \[
        P_{\in T_f(S)}=M_f\sum_{G\in S}\left(\prod_{(i,j)\in E(G)}a_{2+ni+j,1}\prod_{(i,j)\in \overline{E(G)}}a_{2+ni+j,0}\right),
    \]
    and
    \[
        P_{\in T(S)}=\sum_{G\in S}\left(\prod_{(i,j)\in E(G)}a_{ni+j,1}\prod_{(i,j)\in \overline{E(G)}}a_{ni+j,0}\right).
    \]
    The Chow rank of $P_{f\in T(S)}$ is at least the Chow rank of $P_{\in S}$ since $P_{f\in S}$ is the restriction of $P_{\in T(S)}$ along $a_{0,0}=a_{1,0}=a_{2+ni+j,0}=1$ followed by the relabelling $a_{2+ni+j,1} \rightarrow a_{i,j}$ for all $(i,j)\in\Z_n\times\Z_n$. Similarly, $P_{\in S}$ is the restriction of $P_{\in T(S)}$ along $a_{ni+j,0}=1$ followed by the relabelling $a_{ni+j,1}\rightarrow a_{i,j}$ for all $(i,j)\in\Z_n\times\Z_n$, and thus the claim holds. 
\end{proof}

\begin{ex}
    Again consider the graph $G$ on two vertices with monomial edge listing $a_{00}a_{01}$. The images of $M_G$ under $T_{a_{00}a_{10}}$ and $T$ are
    \[
        T_{a_{00}a_{10}}(G)=\left(a_{0,0}a_{1,0}\prod_{(i,j)\in E(G)}a_{2+ni+j,1}\prod_{(i,j)\in \overline{E(G)}}a_{2+ni+j,0}\right)=a_{00}a_{10}a_{21}a_{31}a_{40}a_{50},
    \]
    and
    \[
        T(G)=\left(\prod_{(i,j)\in E(G)}a_{ni+j,1}\prod_{(i,j)\in \overline{E(G)}}a_{ni+j,0}\right)=a_{01}a_{11}a_{20}a_{30},
    \]
    which underlie the functional graphs depicted in Figure \ref{fig:trans}.

    \begin{figure} \begin{center} $\begin{array}{ccccc}\begin{tikzpicture}
    \node (0) at (0,0) {};
    \node (1) at (1.5,0) {};
    \draw[fill=black] (0,0) circle (3pt);
    \draw[fill=black] (1.5,0) circle (3pt);
    \node at (0,-0.5) {$0$};
    \node at (1.5,-0.5) {$1$};
    \draw (0) edge[color=magenta, very thick,->,out=135,in=45,looseness=10, ] (0) node [above=20pt,fill=white] {$\color{magenta} a_{00}$};
     \draw (0) edge[color=teal, very thick,->,out=315,in=225,looseness=1] (1) node [below=22pt,right=16pt,fill=white] {$\color{teal}a_{01}$};
    \end{tikzpicture} & \quad\quad &

    \begin{tikzpicture}
    \node (0) at (1.5,2) {};
    \node (1) at (3,2) {};
    \node (2) at (0,0) {};
    \node (3) at (1.5,0) {};
    \node (4) at (3,0) {};
    \node (5) at (4.5,0) {};
    \draw[fill=black] (1.5,2) circle (3pt);
    \draw[fill=black] (3,2) circle (3pt);
    \draw[fill=black] (0,0) circle (3pt);
    \draw[fill=black] (1.5,0) circle (3pt);
    \draw[fill=black] (3,0) circle (3pt);
    \draw[fill=black] (4.5,0) circle (3pt);
    \node at (1.5,2.5) {$0$};
    \node at (3,2.5) {$1$};
    \node at (0,-0.5) {$2$};
    \node at (1.5,-0.5) {$3$};
    \node at (3,-0.5) {$4$};
    \node at (4.5,-0.5) {$5$};
    \node at (-0.5,0) {$\color{magenta}a_{00}$};
    \node at (1,0) {$\color{teal}a_{01}$};
    \node at (2.5,0) {$a_{10}$};
    \node at (4,0) {$a_{11}$};
    \draw (0) edge[thick, ->, in=135, out=225, looseness=10] (0);
    \draw (1) edge[thick, ->, in=45, out=135, looseness=1] (0);
    \draw (2) edge[color=magenta, very thick,->,out=90,in=270,looseness=0.5] (1);
    \draw (3) edge[color=teal, very thick,->,out=90,in=270,looseness=0.5] (1);
    \draw (4) edge[thick,->,out=90,in=270,looseness=0.5] (0);
    \draw (5) edge[thick,->,out=90,in=270,looseness=0.5] (0);
    \end{tikzpicture} & \quad\quad &
    
    \begin{tikzpicture}
    \node (0) at (0,0) {};
    \node (1) at (1.5,0) {};
    \node (2) at (3,0) {};
    \node (3) at (4.5,0) {};
    \draw[fill=black] (0,0) circle (3pt);
    \draw[fill=black] (1.5,0) circle (3pt);
    \draw[fill=black] (3,0) circle (3pt);
    \draw[fill=black] (4.5,0) circle (3pt);
    \node at (0,-0.5) {$0$};
    \node at (1.5,-0.5) {$1$};
    \node at (3,-0.5) {$2$};
    \node at (4.5,-0.5) {$3$};
    \node at (-0.5,0) {$\color{magenta}a_{00}$};
    \node at (1,0) {$\color{teal}a_{01}$};
    \node at (2.5,0) {$a_{10}$};
    \node at (4,0) {$a_{11}$};
    \draw (0) edge[color=magenta, very thick,->,out=315,in=225,looseness=1] (1);
    \draw (1) edge[color=teal, very thick,->,out=135,in=45,looseness=10] (1);
    \draw (2) edge[thick,->,out=90,in=90,looseness=1] (0);
    \draw (3) edge[thick,->,out=90,in=90,looseness=1] (0);
    \end{tikzpicture}
    \end{array}$ \end{center} \caption{Example of the transformations in \ref{lem:trans} showing $G$, $T_{a_{00}a_{10}}(G)$, and $T(G)$ from left to right.} \label{fig:trans}
    \end{figure}
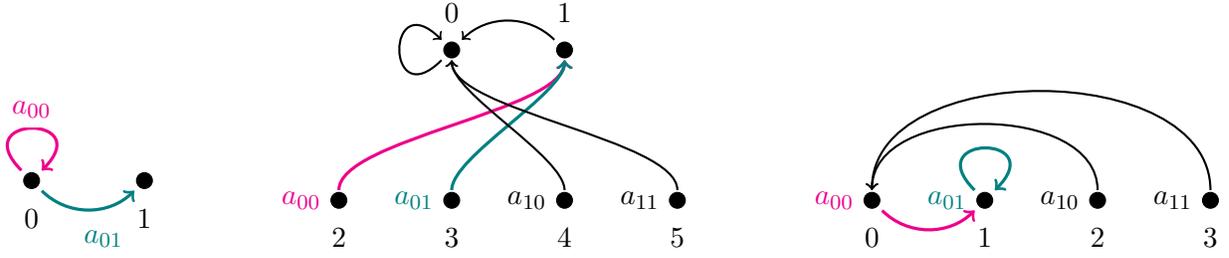
\end{ex}

\begin{defn}
    \textbf{Functional computers} are the restriction of differential computers to functional inputs and additive listings of sets of functions from $\Z_n^{\Z_n}$. Explicitly, given an additive listing $P_{\in S,m}$ of $S\subseteq\Z_n^{\Z_n}$, a \textbf{functional computer} $F_{\in S}:\Z_n^{\Z_n}\rightarrow\{0,1\}$ is the following implementation of the Boolean function $\F_{\in S}$.
    \[
        F_{\in S}(g)=\lp\left(\underset{i\in\Z_{n}}{\prod}\frac{\partial}{\partial a_{i,g(i)}}\right)P_{\in S,m}(A)\rp^{m}
    \]
\end{defn}

\begin{lem} \label{lem:ac}
    Every functional computer of Chow rank $\rho$ in $n$ variables can be implemented by the arithmetic formula
    \[
        A_{\rho,n}(X)=\sum_{u\in\Z_\rho}\prod_{v\in\Z_n}x_{u,v}
    \]
\end{lem}

\begin{proof}
    First observe that the additive listing of $S\subseteq\Z_n^{\Z_n}$ must be homogenous of degree $n$. Thus by Lemma \ref{lem:homo}, we can restrict our attention to homogenous Chow decompositions of degree $n$.
    Then given such a Chow decomposition of $P_{\in S,m}$,
    \[
        F_{\in S}(g)=\lp\left(\underset{i\in\Z_n}\prod\frac{\partial}{\partial a_{i,g(i)}}\right)\sum_{u\in\Z_\rho}\prod_{v\in\Z_n}\sum_{w\in\Z_n}H_{u,v,w}a_{v,w}\rp^m
    \]
    \[
        =\lp\sum_{u\in\Z_\rho}\prod_{v\in\Z_n}\sum_{w\in\Z_n}\frac{\partial}{\partial a_{v,g(v)}}H_{u,v,w}a_{v,w}\rp^m
        = \lp\sum_{u\in\Z_\rho}\prod_{v\in\Z_n}H_{u,v,g(v)}\rp^m
    \]
    Next observe that the depth\textendash two arithmetic formula
    \[
        A_{\rho,n}(X)=\sum_{u\in\Z_\rho}\prod_{v\in\Z_n}x_{u,v}
    \]
    satisfies $A_{\rho,n}(X)=0\Leftrightarrow F_{\in S}(g)=0$ for inputs $x_{u,v}=H_{u,v,g(v)}$. Thus $A_{\rho,n}(X)$ can be used to implement $F_{\in S}$.
\end{proof}

\section{Totally Non\textendash Overlapping Polynomials}

\subsection{Monomial Non\textendash Overlapping Lemma}

For the duration of this section we consider polynomials with $n\in\Z$ additive terms.

\begin{defn}
    For $m\in\Z$, $m\geq2$ we define $P_m\in\C[x_0,...,x_{mn-1}]$ to be the following homogenous $m$\textendash linear polynomial which has \textbf{totally non\textendash overlapping} monomial support and $\alpha_i\neq 0$ for all $i\in\Z_n$.
    $$P_m=\sum_{i\in\Z_n}\alpha_i \prod_{j\in\Z_m}{x_{mi+j}}$$
    Observe $P_m$ has $n$ additive terms and multilinear degree $m$.
\end{defn}

\begin{ex}
    For $n=3$,
    \[
        P_3=x_0x_1x_2+x_3x_4x_5+x_6x_7x_8.
    \]
    Intuition tells us this totally non\textendash overlapping polynomial should have Chow rank $n$, since no additive terms have common factors. The remainder of the subsection is dedicated to a proof of this fact.
\end{ex}

\begin{lem} \label{lem:pm}
    For all $m\in\Z$, the Chow rank of $P_m$ is at least the Chow rank of $P_2$.
\end{lem}
\begin{proof}
    Observe that $P_2$ is the restriction of $P_m$ along $x_{mi+j}=1$ for all $i\in\Z_n$, $j\in\Z_m\setminus\{0\}$, which can only decrease the Chow rank.
\end{proof}

\begin{lem} \label{lem:homo}
    For all homogenous $P\in\C[x_0,...,x_{n}]$ of degree $d$, there exists an optimal Chow\textendash decomposition of $P$ which is homogenous.
\end{lem}
\begin{proof}
    Suppose $P$ has optimal Chow\textendash decomposition
    \[
        P=\sum_{u\in\Z_\rho} \prod_{v\in\Z_d}\left( H_{u,v,n} + \sum_{w\in\Z_n}{H_{u,v,w}x_w}\right).
    \]
    Since $P$ is homogenous of degree $d$, all terms in its expanded form must have degree $d$. However, any term containing a factor $H_{u,v,n}$ can have degree at most $d-1$, so all such terms must cancel. Thus substituting $H_{u,v,n}=0$ for all $u\in\Z_\rho$, $v\in\Z_d$ yields another optimal Chow-decomposition which is homogenous.
\end{proof}

\begin{lem} \label{lem:skew}
    Homogenous Chow\textendash decompositions of degree two correspond to complex skew\textendash symmetric matrices.
\end{lem}
\begin{proof}
    Suppose $P\in\C[x_0,...,x_{n-1}]$ is homogenous of degree two. Then there exists a unique symmetric $A\in\C^{n\times n}$ such that $P=x^TAx$. On the other hand, a homogenous Chow-decomposition of $P$ has the form
    \[
        E_2=\sum_{0 \leq u < \rho} \prod_{0 \leq v <2}\sum_{w\in\Z_n}H_{u,v,w}x_w
    \]
    \[
        =\sum_{0 \leq u < \rho} \left(\sum_{w\in\Z_n}H_{u,0,w}x_w\right)\left( \sum_{w\in\Z_n}H_{u,1,w}x_w\right)
    \]
    \[
        =\sum_{0 \leq u < \rho} \left(\boldsymbol{H}^\top _{u,0}\boldsymbol{x}\right)\left(\boldsymbol{H}^\top_{u,1}\boldsymbol{x}\right)
        =\sum_{0 \leq u < \rho} \boldsymbol{x}^\top \boldsymbol{H}_{u,0}\boldsymbol{H}^\top _{u,1}\boldsymbol{x}
    \]
    \[
        =\boldsymbol{x}^\top \left(\sum_{0 \leq u < \rho} \boldsymbol{H}_{u,0}\boldsymbol{H}^\top _{u,1}\right)\boldsymbol{x}
        =\boldsymbol{x}^\top B\boldsymbol{x}
    \]
    where $B\in\C^{n\times n}$ is the sum of $\rho$ rank one matrices and thus has rank at most $\rho$.

    Next recall $B$ can be expressed uniquely as the sum of a symmetric and a skew\textendash symmetric matrix $B=B_\text{sym}+B_{skew}$ and that the quadratic form $\boldsymbol{x}^TB\boldsymbol{x}$ is uniquely determined by $B_\text{sym}$. Thus $B_\text{sym}=A$, and the set of matrices $B$ whose quadratic form satisfies $P=\boldsymbol{x}^TB\boldsymbol{x}$ is $\{A+S\,\,|\,\,S=-S^T\}$. The desired correspondence is between Chow decompositions $B$ of $P$ and skew symmetric matrices $S$.
\end{proof}

\begin{cor} \label{cor:rank}
    The Chow rank of a homogenous degree two multilinear polynomial $P\in\C[x_0,...,x_{n-1}]$ is given by
    \[
        \min_{S=-S^\top }\rank\left(A+S\right)
    \]
    where $S\in\C^{n\times n}$, and $A\in\C^{n\times n}$ is the unique symmetric matrix such that $P=\boldsymbol{x}^\top A\boldsymbol{x}$.
\end{cor}
\begin{proof}
    By Lemma \ref{lem:skew}, rank $\rho$ Chow decompositions of $P$ correspond to matrices of the form $B+S$ with rank $\rho$. Conversely, if $P=\boldsymbol{x}^\top  A\boldsymbol{x}$ where $A\in\C^{n\times n}$ has rank $\rho$, then $A$ can be written as the sum of $\rho$ rank one matrices, yielding a rank $\rho$ Chow\textendash decomposition of $P$.
\end{proof}

\begin{lem} \label{lem:e2}
    The Chow rank of $P_2$ is exactly $n$.
\end{lem}
\begin{proof}
    Observe that for $\boldsymbol{x}=(x_0,...,x_{2n-1})^\top $,
    \[
        P_2=\boldsymbol{x}^\top \left(\frac{\alpha_i}{2}\dsum_{i\in\Z_n}
        \begin{bmatrix}
            0   &   1   \\
            1   &   0
        \end{bmatrix}
        \right)\boldsymbol{x}.
    \]
    Thus by Corollary \ref{cor:rank}, the Chow rank of $P_2$ is given by
    \[
        \min_{S=-S^\top }\rank\left[\left(\dsum_{i\in\Z_n}
        \begin{bmatrix}
            0   &   1   \\
            1   &   0
        \end{bmatrix}
        \right)+S\right]
    \]
    where $S\in\C^{2n\times 2n}$. Recall that every skew-symmetric matrix $S$ of even size has spectrum $\sigma(S)=\{\pm\lambda_i\}_{i\in\Z_n}$ and can be decomposed as
    \[
        S=U^*\left(\dsum_{i\in\Z_n}
        \begin{bmatrix}
            0   &   \lambda_i \\
            -\lambda_i & 0
        \end{bmatrix}
        \right)U
    \]
    where $\lambda_i\in\C$ and $U\in\C^{n\times n}$ is unitary. Without loss of generality, suppose $U$ is full rank, so
    \[
        \rank\left[\left(\dsum_{i\in\Z_n}
        \begin{bmatrix}
            0   &   1   \\
            1   &   0
        \end{bmatrix}
        \right)+S\right]=\rank\left(\dsum_{i\in\Z_n}
        \begin{bmatrix}
            0   &   1+\lambda_i \\
            1-\lambda_i & 0
        \end{bmatrix}\right)
        =\sum_{i\in\Z_n}\rank\left(
        \begin{bmatrix}
            0   &   1+\lambda_i \\
            1-\lambda_i & 0
        \end{bmatrix}\right)
    \]
    which by inspection has a global minimum of $n$ when $\lambda_i=0$ for all $i\in\Z_n$, i.e. $S=\mathbb{0}$.
\end{proof}

\begin{lem}
    The Chow\textendash rank of an $n$ term homogenous multilinear polynomial $P$ with total degree at least two and totally non\textendash overlapping monomial support is exactly $n$.
\end{lem}
\begin{proof}
    Observe that the variables of $P$ can be relabelled to yield $P_m$ without affecting the Chow rank. Then by Lemma \ref{lem:pm}, the Chow rank of $P_m$ is at least the Chow\textendash rank of $P_2$, which is exactly $n$. The trivial Chow decomposition of $P_m$ into $n$ Chow rank one summands matches the lower bound, proving the desired monomial non\textendash overlapping lemma.
\end{proof}

\subsection{Applications to Totally Non\textendash Overlapping Polynomials}

We now use the monomial non\textendash overlapping lemma to prove lower bounds on the Chow ranks of two polynomials of interest.

\begin{cor}
    Let $P_{=c}$ denote the additive listing of the set of identically constant functions $S_{=c} \subset {\Z_n}^{\Z_n}$. Then the Chow rank of $P_{=c}$ is exactly $n$.
\end{cor}
\begin{proof}
    First write
    \[
        P_{=c}\lp A\rp=\sum_{j\in\Z_n}\prod_{i\in\Z_n}a_{i,j},
    \]
    and observe that the polynomial is totally non\textendash overlapping. That is, for each $(u,v)\in\Z\times\Z$, the variable $a_{u,v}$ appears in $P_{=c}$ exactly once, when $i=u$ and $j=v$. Thus by the monomial non\textendash overlapping lemma the claim holds.
\end{proof}

\begin{cor}
    Let $P_{C_n}$ denote the additive listing of the subgroup $\text{C}_{n}\cong \nicefrac{\mathbb{Z}}{n\mathbb{Z}} \subset \text{S}_n$ generated by $f_{+1}\in\Z_n^{\Z_n}$ where $f_{+1}(x)=x+1\text{ mod }n$. Then the Chow rank of $P_{C_n}$ is exactly $n$.
\end{cor}
\begin{proof}
    First write
    \[
        P_{\text{C}_n}\lp A\rp=\sum_{\sigma\in\text{C}_{n}}\prod_{i\in\mathbb{Z}_{n}}a_{i,\sigma(i)}=\sum_{j\in\Z_n}\prod_{i\in\Z_n}a_{i,i+j \text{ mod }n},
    \]
    and observe that the polynomial is totally non\textendash overlapping. That is, for each $(u,v)\in\Z\times\Z$, the variable $a_{u,v}$ appears in $P_{\text{C}_n}$ exactly once, when $i=u$ and $j=u+v \text{ mod }n$. Thus by the monomial non\textendash overlapping lemma the claim holds.
\end{proof}




\section*{Acknowledgments}
This work was supported in part by the United States Office of Naval Research, Code 321.\\
We also thank Joshua Grochow, Harm Derksen, Joseph
Landsberg, Matthew Santana, and Michael Williams for enlightening conversations.

\bibliographystyle{plain}
\bibliography{main.bib}

\end{document}